\documentclass[11pt]{article}
\usepackage[latin1]{inputenc}
\usepackage{amsmath, color}
\usepackage{amsfonts}
\usepackage{amssymb}
\usepackage[left=2.50cm, right=2.50cm, top=2.50cm, bottom=2.50cm]{geometry}
\usepackage{makeidx}
\usepackage{graphicx}
\usepackage{hyperref}
\usepackage{amsthm, tikz}
\usepackage{multirow}
\usepackage{multicol}
\usepackage{longtable}

\usepackage[light,condensed,math]{kurier}
\usepackage[T1]{fontenc}

\usepackage{mathpazo}      % Math
\usepackage{charter}        % Text Charter
\newcommand{\FF}{{\mathbb{F}}}

\newcommand{\seq}{\subseteq}
\newcommand{\ZZ}{\mathbb{Z}}

\newcommand{\av}{\mathbf{a}}
\newcommand{\bv}{\mathbf{b}}
\newcommand{\cv}{\mathbf{c}}

\newcommand{\xv}{\mathbf{x}}
\newcommand{\yv}{\mathbf{y}}
\makeatother
\usepackage{etoolbox}
\apptocmd\normalsize{%
    \abovedisplayskip=9pt plus 3pt minus 6pt
    \abovedisplayshortskip=3pt plus 3pt
    \belowdisplayskip=9pt plus 3pt minus 9pt
    \belowdisplayshortskip=6pt plus 3pt minus 4pt
}{}{}

\theoremstyle{theorem}
\newtheorem{theorem}{Theorem}[section]
\newtheorem{lem}[theorem]{Lemma}
\newtheorem{cor}[theorem]{Corollary}

\theoremstyle{definition}
\newtheorem{ex}[theorem]{Example}
\newtheorem{definition}{Definition}[]
\newtheorem{rem}{Remark}

\title{Skew cyclic codes over $\ZZ_4+v\ZZ_4$ with derivation: structural properties and computational results}

\author{Djoko Suprijanto
\footnote{Combinatorial Mathematics Research Group, Faculty of Mathematics and Natural Sciences,
Institut Teknologi Bandung,
Jl. Ganesha 10, Bandung, 40132,
INDONESIA,\hfill{djoko.suprijanto@itb.ac.id}}
\hspace{0.1cm}and Hopein Christofen Tang
\footnote{Combinatorial Mathematics Research Group, Faculty of Mathematics and Natural Sciences, Institut Teknologi Bandung, Jl. Ganesha 10, Bandung, 40132, INDONESIA.	{\bf Current address:} School of Mathematics and Statistics, UNSW, Sydney, Australia, \hfill \texttt{hopein.tang@unsw.edu.au}}
}
\date{}

\begin{document}

\maketitle

\begin{abstract}
In this work, we study a class of skew cyclic codes over the ring $R:=\ZZ_4+v\ZZ_4,$ where $v^2=v,$ with an automorphism $\theta$
and a derivation $\Delta_\theta,$ namely codes as modules over a skew polynomial ring  $R[x;\theta,\Delta_{\theta}],$ whose multiplication is defined using an automorphism $\theta$ and a derivation $\Delta_{\theta}.$  We investigate the structures of a skew polynomial ring $R[x;\theta,\Delta_{\theta}].$  We define $\Delta_{\theta}$-cyclic codes as a generalization of the notion of cyclic codes.  The properties of $\Delta_{\theta}$-cyclic codes as well as dual $\Delta_{\theta}$-cyclic codes are derived.  As an application, some new linear codes over $\ZZ_4$ with good parameters are obtained by Plotkin sum construction, also via a Gray map as well as residue and torsion codes of these codes.\\

\end{abstract}
	

\section{Introduction}

Cyclic codes are an important class of codes from both theoretical and practical viewpoints. Theoretically, cyclic codes have a rich mathematical theory, in particular, they have additional algebraic structures to make,
practically, the process of encoding and decoding cyclic codes is more efficient.
%\textcolor{red}{In fact, cyclic codes can be generated using any divisor of $x^n-1$ in some commutative polynomial ring}.

Cyclic codes over finite fields were first studied by Prange \cite{Prange1957} in 1957.  Since then, many coding theorists have made significant progress in studying cyclic codes for both the so-called random-error correction and burst-error correction (See, for example, \cite{Huffman} for the detailed description of random-error and burst-error correction).

In 2007, Boucher, Geiselmann, and Ulmer \cite{Boucher2007}  (see also \cite{Boucher2009a},\cite{Boucher2009b}) extended the notion of cyclic codes over finite fields by using generator polynomials in non-commutative skew polynomial rings.  The new notion of codes is called skew cyclic codes over finite fields. In general, a skew polynomial ring is not a unique factorization ring. In this case, there are typically more factors of $x^n-1$ in this ring than in the commutative case.  Hence, there are more skew cyclic codes than cyclic codes over finite fields. This new class of codes increases the number of possibilities for finding codes with good parameters. Boucher, Geiselmann, and Ulmer \cite{Boucher2007} also give many examples of codes that improve the previously best-known linear codes over finite fields.

Later, the notion of skew cyclic codes over finite fields was generalized to the skew cyclic codes over several kinds of finite rings.  Abualrub, Aydin, and Seneviratne  \cite{Abualrub2012} considered skew cyclic codes over $\FF_2+v\FF_2,$ where $v^2=v,$  and constructed optimal self-dual codes over this ring.  Gursoy, Siap, and Yildiz \cite{Gursoy2014} investigated structural properties of skew cyclic codes over $\FF_q+v\FF_q,$ with $v^2=v.$  They \cite{Gursoy2014} showed that skew cyclic codes over the ring are principally generated.  Later, the first author together with his coauthors considered  structural aspects of skew cyclic codes over the rings $A_k$ \cite{Djoko2016} and $B_k$ \cite{Djoko2018} (c.f. \cite{Djoko2018b}), respectively. Very recently, Benbelkacem, Ezerman, Abualrub, Aydin, and Batoul \cite{Ezerman2021} considered the skew cyclic codes over the mixed alphabet which are also a finite ring, denoted by $\FF_4R,$ where $R=\FF_4+v \FF_4$ with $v^2=v.$  They \cite{Ezerman2021} showed a natural connection between the skew cyclic codes over the ring to  DNA codes.

In the next development, Boucher and Ulmer \cite{Boucher2014} generalized the notion of skew cyclic codes over finite fields to the skew cyclic codes over finite fields with derivation.  They \cite{Boucher2014} also constructed MDS as well as MRD codes from certain families of the skew cyclic codes (see Section 4.3 in \cite{Boucher2014}). Sharma and Bhaintwal \cite{Sharma2018} extended the study of these skew cyclic codes over a finite ring, namely over the ring $\ZZ_4+u\ZZ_4,$ with $u^2=1.$ 
They obtained numerous linear codes over $\ZZ_4$ with good parameters using residue codes, Plotkin sum, or the Gray map they defined \cite{Sharma2018}. Very recently, Patel and Prakash \cite{Patel2022} have investigated the same object over the ring $\FF_q+u\FF_q+v\FF_q+uv\FF_q,$ with $u^2=u$ and $v^2=v,$ where $q$ is a prime power.  By the decomposition method, they \cite{Patel2022} obtained several optimal linear codes over $\FF_q.$

Continuing the study of \cite{Sharma2018} and \cite{Patel2022}, in this paper, we investigate a class of skew cyclic codes with derivation over the ring $\ZZ_4+v\ZZ_4,$ with $v^2=v.$  We derive several structural properties of skew cyclic codes with derivation over the ring $\ZZ_4+v\ZZ_4.$  As a by-product, we construct many new linear codes over $\ZZ_4$ with good parameters.

The organization of the paper is as follows.  In Section 2, we provide some definitions and basic facts related to the ring $\ZZ_4+v\ZZ_4$ and also the linear codes over the ring  $\ZZ_4+v\ZZ_4.$ We also define a Gray map from $\ZZ_4+v\ZZ_4$ to $\ZZ_4^2,$ which can be extended naturally to define the Gray map from $(\ZZ_4+v\ZZ_4)^n$ to $\ZZ_4^{2n}.$  Several properties of the skew-polynomial ring $(\ZZ_4+v\ZZ_4)[x;\theta,\Delta_{\theta}]$ are derived.  The notion of $\Delta_{\theta}$-cyclic codes, as well as dual of $\Delta_{\theta}$-cyclic codes as a generalization of cyclic codes together with their properties are investigated in Section 3 and Section 4, respectively. Several examples of linear codes over $\ZZ_4$ with good parameters obtained by using the Plotkin sum construction, a Gray map, or residue and torsion codes of these classes of codes are provided in Section 5.  The paper is ended with concluding remarks.
We follow \cite{Huffman} for undefined terms in coding theory.
	
\section{Preliminaries}

In this section, we present some definitions together with some basic facts regarding the ring $\ZZ_4+v\ZZ_4,$ linear codes over the ring, and the skew-polynomial ring $(\ZZ_4+v\ZZ_4)[x;\theta,\Delta_{\theta}]$ that is required in the next sections.

\subsection{The ring $\ZZ_4+v\ZZ_4$}
Let $R:=\ZZ_4+v\ZZ_4=\{a+bv:~a,b \in \ZZ_4\},$ with $v^2=v.$  This ring is isomorphic to a polynomial ring, namely $R \cong \displaystyle \frac{\ZZ_4[v]}{\langle v^2-v \rangle}.$  An element $a+bv \in R$ is a unit if and only if $a$ and $a+b$ both are units in $\ZZ_4$.  Since the units of $\ZZ_4$ are $1$ and $3,$ the units of $R$ are $1,$ $3,$ $1+2v,$ and $3+2v.$  Hence, the non-units of $R$ are
\[
\{0,~2,~v,~2v,~3v,~1+v,~1+3v,~2+v,~2+2v,~2+3v,~3+v,~3+3v\}.
\]
$R$ is a principal ideal ring with $7$ non-trivial ideals, namely
\[
\begin{aligned}
	\langle 2v \rangle &=\{0,2v\},\\
	\langle 2+2v \rangle &=\{0,2+2v\},\\
   	\langle 2 \rangle &=\{0,2,2v,2+2v\},\\
	\langle v\rangle &=\{0,v,2v,3v\}=\langle 3v \rangle,\\
	\langle 3+v \rangle &=\{0,1+3v,2+2v,3+v\}=\langle 1+3v\rangle,\\
	\langle 1+v \rangle &=\{0,2,2v,1+v,1+3v,2+2v,3+v,3+3v\}=\langle 3+3v\rangle,\\
	\langle 2+v \rangle &=\{0,2,v,2v,3v,2+v,2+2v,2+3v\}=\langle 2+3v\rangle.\\
\end{aligned}
\]
The maximal ideals of $R$ are $\langle 1+v \rangle$ and $\langle 2+v \rangle$.  Hence, $R$ is a semi-local ring.
%The lattice diagram (with respect to inclusion) of ideals is given in Figure~\ref{G1}.
%
%
%	\begin{figure}[h!]
%	\begin{center}	
%	\begin{tikzpicture}
%	\node (top) at (0,2) {$R$};
%	\node (a) at (-2,0) {$\langle 1+v \rangle$};
%	\node (b) at (2,0) {$\langle 2+v \rangle$};
%	\node (c) at (2,-2) {$\langle 2 \rangle$};
%	\node (d) at (4,-2) {$\langle v \rangle$};
%	\node (e) at (-2,-4) {$\langle 3+v \rangle$};
%	\node (f) at (2,-4) {$\langle 2+2v \rangle$};
%	\node (g) at (4,-4) {$\langle 2v \rangle$};
%	\node (h) at (0,-6) {$\langle 0 \rangle$};
%	\draw (h)--(e);
%	\draw (h)--(f);
%	\draw (h)--(g)--(d);
%	\draw (e)--(a)--(f)--(c)--(g);
%	\draw (a)--(top)--(b)--(c);
%	\draw (b)--(d);
%	\end{tikzpicture}
%    \end{center}
%    \caption{Lattice diagram of ideals in the ring $R$}\label{G1}
%    \end{figure}

For more information on the structures of the ring $R=\ZZ_4+v\ZZ_4,$ the reader can refer to \cite{Bandi2004}, \cite{Gao2017}, and \cite{Liu2020}.

\subsection{Linear codes over $R$}

Lee weight is an important weight to consider on $\ZZ_4$.  For $x \in \ZZ_4,$ the Lee
weight of $x,$ denoted by $w_L(x),$ is defined as $w_L(0)=0,$ $w_L(1)=1=w_L(3),$  $w_L(2)=2.$  The Lee weight for any vector $(r_0,r_1,\ldots,r_{n-1}) \in \ZZ_4^n$ is defined as the rational sum of Lee weights of its coordinates, namely $w_L((r_0,r_1,\ldots,r_{n-1}))=w_L(r_0)+w_L(r_1)+\cdots+w_L(r_{n-1}).$
%\textcolor{red}{(definisi Lee weight hanya di  }$\color{red}{\ZZ_4^2}$ \textcolor{red}{ saja?)}

Define a Gray map $\phi:R \longrightarrow \ZZ_4^2$ as
\[
\phi(a+bv)=(a,a+b).
\]
The Gray weight $w_G(a+bv)$ for any $a+bv \in R$ is defined as $w_G(a+bv)=w_L(\phi(a+bv)).$  The Gray weights of the elements of $R$ are given as follows.
\[
\begin{array}{c|cccccccc}
	\hline\hline
	x & 0 & 1 & 2 & 3 & v & 2v & 3v & 1+v \\
	w_G(x) & 0 & 2 & 4 & 2 & 1 & 2 & 1 & 3\\
	\hline
	x & 1+2v & 1+3v & 2+v & 2+2v & 2+3v & 3+v & 3+2v & 3+3v \\
	w_G(x) 	& 2 & 1 & 3 & 2 & 3 & 1 & 2 & 3\\
	\hline\hline
\end{array}
\]
The Gray map $\phi$ is extended naturally to $\Phi:R^n \longrightarrow \ZZ_4^{2n}$ as
\[
\Phi((a_0+b_0v,a_1+b_1v,\ldots,a_{n-1}+b_{n-1}v))=
(a_0,a_0+b_0,a_1,a_1+b_1,\ldots,a_{n-1},a_{n-1}+b_{n-1}),
\]
and the Gray weight of any vector $\xv \in R^n$ is defined as the rational sum of Gray weights of its coordinates.

A code $C$ of length $n$ over $R$ is a non-empty subset of $R^n.$  A code $C$ is called linear over $R$ if it is an $R$-submodule of $R^n.$ A linear code over $R$ is called free if it is free as an $R$-submodule. The Gray distance of any vectors $\xv,\yv \in R^n$ is defined as $d_G(\xv,\yv)=w_G(\xv-\yv)$ and the Lee distance of any vectors $\xv,\yv \in \ZZ_4^n$ is defined as $d_L(\xv,\yv)=w_L(\xv-\yv)$.  The minimum Gray distance $d_G(C)$ and the minimum Lee distance $d_L(C)$ of $C$ is defined as
$d_G(C):=\text{min}\{d_G(\xv,\yv):~\xv,\yv \in C,~\xv \neq \yv\}$ and
$d_L(C):=\text{min}\{d_L(\xv,\yv):~\xv,\yv \in C,~\xv \neq \yv\},$ respectively.
It is easy to verify that the Gray map $\Phi$ is a distance-preserving map (isometry) from $(R^n,d_G)$ to $(\ZZ_4^{2n},d_L).$

A (linear) code over the ring $\ZZ_4$ is defined similarly.  We write the parameters of a linear code $C$ over $\ZZ_4$ as $[n,4^{k_1}2^{k_2},d_L],$ where $n$ is the length of $C,$ $|C|=4^{k_1}2^{k_2},$ and $d_L=d_L(C).$  Moreover, following Hammons, Kumar, Calderbank, Sloane, and Sol\'{e}  \cite{Hammons1994} (cf. \cite{Wan}), we say that the code $C$ is of type $4^{k_1}2^{k_2}.$

%The upper bound for the Lee distance of the linear codes over the ring %$\ZZ_4$ was proved by Dougherty and Shiromoto \cite{Dougherty2001}. %\textcolor{red}{Apakah bound ini perlu disebut? Kalau di paper Sharma memang %perlu karena istilah MLDS digunakan di kode yang didapat, sedangkan kode yang %didapat di paper ini istilah MLDS dan bound tsb tidak terpakai sama sekali.}

%\begin{theorem}[Lee distance bound]\label{Lee-bound}
%	If $C$ is a linear code of length $n$ over $\ZZ_4$ with parameters %$[n,4^{k_1}2^{k_2},d_L],$ then $d_L \leq 2n-2k_1-k_2+1.$
%\end{theorem}

%A linear code $C$ over $\ZZ_4$ which attains the Lee distance bound is called %a Maximum Lee Distance Separable (MLDS) code.
%MLDS codes are a kind of extremal codes which are investigated and %constructed by many coding theorists.

For a linear code $C \seq R^n$ over $R,$ we define the residue code $Res(C)$ and the torsion code $Tor(C)$ of $C,$ respectively, as
\[
Res(C):=\{\av:~\av+\bv v \in C, \text{ for some }\bv \in \ZZ_4^n\},
\]
and
\[
Tor(C):=\{\bv:~\bv v \in C\}.
\]

We note that $Res(C)$ and $Tor(C)$ are linear codes of length $n$ over $\ZZ_4.$

%If $\av+\bv v \in C,$ then we have $(\av+\bv v)v=(\av+\bv)v \in C.$  It %implies $(\av+\bv) \in Tor(C).$ \textcolor{red}{sifat ini tidak terpakai?}

Regarding the residue and torsion codes, we have the following property.

\begin{lem}\label{Res-1}
	Let $C$ be a free linear code of length $n$ over $R$ and $\{\cv_1,\cv_2,\ldots,\cv_k\},$ with $\cv_i=\av_i+\bv_iv,$ be a basis of $C.$ Then the following statements hold:
	\begin{itemize}
		\item[(1)] $Res(C)$ is a free linear code over $\ZZ_4$ with  $\{\av_1,\av_2,\ldots,\av_k\}$ as a basis.
		
		\item[(2)] $Tor(C)$ is a free linear code over $\ZZ_4$ with $\{\av_1+\bv_1,\av_2+\bv_2,\ldots,\av_k+\bv_k\}$ as a basis.
	\end{itemize}
\end{lem}

\begin{proof}
%	\begin{itemize}
We prove only part (1).	
%\item[(1)]	
	Let $\av \in Res(C).$  Then there exists $\bv \in \ZZ_4^n$ such that $\av+\bv v \in C.$  Since $\{\cv_1,\cv_2,\ldots,\cv_k\}$ is a basis of $C,$ there exist $r_1,r_2,\ldots,r_k \in R,$ with $r_i=s_i+t_iv,$ for $i \in [1,k]_\ZZ,$ such that
%	\[
%	\begin{aligned}
%	\av+\bv v &=\sum_{i=1}^{k}r_i \cv_i\\
%	          &=\sum_{i=1}^{k}(s_i+t_iv)(\av_i+\bv_iv)\\
%	          &=\sum_{i=1}^{k}s_i\av_i+\left(\sum_{i=1}^{k}(s_i\bv_i+t_i\av_i+t_i\bv_i) \right)v.
%	\end{aligned}
%	\]
\[
\av+\bv v =\sum_{i=1}^{k}s_i\av_i+\left(\sum_{i=1}^{k}(s_i\bv_i+t_i\av_i+t_i\bv_i) \right)v,
\]
which means $\displaystyle \av=\sum_{i=1}^{k}s_i\av_i,$  and  $\{\av_1,\av_2,\ldots,\av_k\}$ generates $Res(C).$
%Now, suppose on the contrary that $\{\av_1,\av_2,\ldots,\av_k\}$ is linearly dependent. Then there exist $m_1,m_2,\ldots,m_k \in \ZZ_4,$ not all zero, such that
%\[
%\mathbf{0}=\sum_{i=1}^{k}m_i\av_i.
%\]
%But,
%\[
%\begin{aligned}
%	\sum_{i=1}^{k}(m_i+3m_iv)\cv_i &=\sum_{i=1}^{k}(m_i+3m_iv)(\av_i+\bv_iv)\\
%	&=\sum_{i=1}^{k}m_i\av_i+(m_i\bv_i+3m_i\av_i+3m_i\bv_i)v\\
%	&=\sum_{i=1}^{k}m_i\av_i+3\left(\sum_{i=1}^{k}m_i\av_i\right)v\\
%	&=\mathbf{0}.
%\end{aligned}
%\]
%Since $m_1,m_2,\ldots,m_k$ not all are zeros, then  $m_1+3m_1v, m_2+3m_2v, \ldots, m_k+3m_kv$ not all are zeros, contradict to the fact that $\{\cv_1,\cv_2,\ldots,\cv_k\}$ is a basis for $C.$  We conclude that $\{\av_1,\av_2,\ldots,\av_k\}$ is a basis of $Res(C).$

Next, suppose on the contrary that $\{\av_1,\av_2,\ldots,\av_k\}$ is linearly dependent. It is easy to prove that $\{\cv_1,\cv_2,\ldots,\cv_k\}$ is also linearly dependent, a contradiction. Therefore, $\{\av_1,\av_2,\ldots,\av_k\}$ must be linear-ly independent and we conclude that $\{\av_1,\av_2,\ldots,\av_k\}$ is a basis of $Res(C).$
	
%\item[(2)]  Similar to (1).
	
%\end{itemize}
\end{proof}

%Now, let $C$ be a free linear code of length $n$ over $R,$ with a basis $\{\cv_1,\cv_2,\ldots,\cv_k\}.$ For $\cv^\prime \in \Phi(C),$ there exists $\cv \in C$ such that $\cv^\prime=\Phi(\cv).$ Let $\displaystyle \cv=\sum_{i=1}^{k} r_i\cv_i \in C,$ with some $r_i=s_i+t_iv,$ for $i \in [1,k]_\ZZ.$  We have
%\[
%\begin{aligned}
%	\Phi(\cv)&=\Phi\left(\sum_{i=1}^{k} r_i\cv_i\right)\\
%	        &=\sum_{i=1}^{k} \Phi((s_i+t_iv)\cv_i)\\
%	        &=\sum_{i=1}^{k}s_i\Phi(\cv_i)+t_i\Phi(\cv_iv).
%\end{aligned}
%\]
%It means that $\{\Phi(\cv_1),\Phi(\cv_1v),\Phi(\cv_2),\Phi(\cv_2v),\ldots,\Phi(\cv_k),\Phi(\cv_kv)\}$ generates $\Phi(C).$  Next, suppose that $\{\Phi(\cv_1),\Phi(\cv_1v),\Phi(\cv_2),\Phi(\cv_2v),\ldots,\Phi(\cv_k),\Phi(\cv_kv)\}$ is linearly dependent.  Then, there exist $m_1,m_2,\ldots,$ $m_k,\ell_1,\ell_2,\ldots,\ell_k \in \ZZ_4,$ not all are zeros, such that
%\[
%\begin{aligned}
%	0 &= \sum_{i=1}^{k} m_i\Phi(\cv_i)+\ell_i \Phi(\cv_i v)\\
%	  &= \sum_{i=1}^{k} \Phi((m_i+\ell_i v)\cv_i)\\
%	  &= \Phi \left(\sum_{i=1}^{k}(m_i+\ell_i v)\cv_i \right).
%\end{aligned}
%\]
%The last condition implies that $\displaystyle \sum_{i=1}^{k}(m_i+\ell_i v)\cv_i=0.$   Since not all $m_1,m_2,\ldots,m_k,\ell_1,\ell_2,\ldots,\ell_k$ are zeros, then $\{\cv_1,\cv_2,\ldots,\cv_k\}$ becomes linearly dependent,  a contradiction.  We conclude that  $\{\Phi(\cv_1),\Phi(\cv_1v),$ $\Phi(\cv_2),\Phi(\cv_2v),\ldots,\Phi(\cv_k),\Phi(\cv_kv)\}$ is a basis of $\Phi(C).$  Hence, we have proven the following theorem.

Similarly, we obtain the following theorem.
\begin{theorem}\label{T-Gray-1}
	If $C$ be a free linear code of length $n$ over $R$ having a basis $\{\cv_1,\cv_2,\ldots,\cv_k\},$ then $\{\Phi(\cv_1),\Phi(\cv_1v),$ $\Phi(\cv_2),\Phi(\cv_2v),\ldots,\Phi(\cv_k),\Phi(\cv_kv)\}$ is a basis of a free linear code $\Phi(C).$
\end{theorem}
%
%As a corollary, we have the property below.
%
%\begin{cor}\label{C-Res}
%	Let $C$ be a free $\Delta_\theta$-cyclic code of length $n$ over $R$ \textcolor{red}{with } $\color{red}|C|=16^k?$.  Then the free linear codes $Res(C)$ and $Tor(C)$ have parameter $4^k2^0,$ and the free linear code $\Phi(C)$ has parameter $4^{2k}2^0.$
%\end{cor}
%\textcolor{red}{apakah tidak masalah jika akibat di atas muncul sebelum definisi skew-cyclic code?}
%\begin{rem}
%	The Corollary \ref{C-Res} does not hold for the code over $\ZZ_4+u\ZZ_4,$ with $u^2=1,$ namely the ring considered by Sharma and Bhaintwal \cite{Sharma2018}. \textcolor{red}{untuk kasus Gray map codes sepertinya masih berlaku?}
%\end{rem}

\subsection{Skew-polynomial ring $R[x;\theta,\Delta_{\theta}]$}

We first recall the definition of a derivation on a finite ring $\mathbf{R}$ following \cite{Boucher2014}.

\begin{definition}
Let $\mathbf{R}$ be a finite ring and $\Theta:\mathbf{R} \longrightarrow \mathbf{R}$ be an automorphism on $\mathbf{R}.$
Then a map $\Delta_\Theta:\mathbf{R} \longrightarrow \mathbf{R}$ is called a derivation on $\mathbf{R}$ if the following two
conditions are satisfied:
\begin{itemize}
\item[(i)]  $\Delta_\Theta(x+y)=\Delta_\Theta(x)+\Delta_\Theta(y),$ and

\item[(ii)] $\Delta_\Theta(xy)=\Delta_\Theta(x)y+\Theta(x)\Delta_\Theta(y).$
\end{itemize}
\end{definition}

Let $\mathbf{R}$ be a finite ring with an automorphism $\Theta$ and a derivation $\Delta_\Theta.$   The skew-polynomial ring $\mathbf{R}[x;\Theta,\Delta_\Theta]$ is the set of all polynomials over $\mathbf{R}$ with ordinary addition of polynomials and multiplication defined by
\[
xa:=\Theta(a)x+\Delta_\Theta(a),
\]
for any $a \in \mathbf{R}.$  This multiplication is extended to all polynomials in $\mathbf{R}[x;\Theta,\Delta_\Theta]$ in the usual manner.  This kind of ring was introduced by Ore \cite{Ore1933} in 1933, where $\mathbf{R}$ is the finite field $\FF_q.$  See also McDonald \cite{McDonald1974}.

Consider a map $\theta:R \longrightarrow R$ defined by $\theta(a+bv)=a+b-bv.$
It is easy to see that $\theta$ defines an automorphism of $R.$
Moreover, since for all $a+bv \in R$ we have $\theta^2(a+bv)=a+bv,$ we conclude that the order of $\theta$ is $2.$

%The following lemma shows a derivation on the ring $R.$

\begin{lem}
	For all $x \in R,$ a map $\Delta_\theta:R \longrightarrow R$ such that
	\[
	\Delta_{\theta}(x)=(1+2v)(\theta(x)-x)
	\]
defines a derivation on $R.$	
\end{lem}
\begin{proof}
    Straightforward by definition.
\end{proof}
%Now, for $C=s+tv \in R$ let us define a map $\Delta_\theta^{(C)}:R \longrightarrow R$ %such that
%\[
%\Delta_\theta^{(C)}(a+bv)=bC.
%\]

%\begin{proof}
%Let $a+bv,c+dv \in R.$  We have
%\[
%\begin{aligned}
%\Delta_{\theta}((a+bv)+(c+dv)) &= \Delta_{\theta}((a+c)+(b+d)v)
                               %&= (1+2v)((a+c)+(b+d)-(b+d)v-((a+c)+(b+d)v))\\
%                               =b+d\\
%                               &=(1+2v)(b-2bv)+(1+2v)(d-2dv)\\
%                               &= \Delta_{\theta}(a+bv)+\Delta_{\theta}(c+dv),
%\end{aligned}
%\]
%and
%\[
%\begin{aligned}
%\Delta_{\theta}((a+bv)(c+dv)) &= \Delta_{\theta}(ac+(bc+ad+bd)v)\\
                                 % &= (1+2v)((bc+ad+bd)-2(bc+ad+bd)v)\\
 %                                 &=(bc+ad+bd)\\
                                 % &= b(c+dv)+(a+b-bv)d\\
 %                                 &=(1+2v)(b-2bv)(c+dv)
%                                  +(a+b-bv)(1+2v)(d-2dv)\\
%                                  &=\Delta_{\theta}(a+bv)(c+dv)+\theta(a+bv)\Delta_{\theta}(c+dv),
%\end{aligned}
%\]
%and hence, by definition, $\Delta_{\theta}$ is a derivation on $R.$
%\end{proof}

We prove several properties related to the derivation $\Delta_\theta$ on $R.$  We begin with the following that can be derived by a routine computation.

\begin{lem}\label{Property D1}
Let $\Delta_{\theta}(x)=(1+2v)(\theta(x)-x)$  be a derivation on $R.$  Then the following statements hold:
\begin{itemize}
\item[(1)]  $\Delta_{\theta} \theta+\theta \Delta_{\theta} \equiv 0.$

\item[(2)]  $\Delta_{\theta} \Delta_{\theta} \equiv 0.$

\item[(3)]  for all $x \in R,$ $\Delta_{\theta}(x)=0$ $\iff$ $\theta(x)=x.$
\end{itemize}
\end{lem}

%\begin{proof} Let $\Delta_{\theta}(x)=(1+2v)(\theta(x)-x)$  be a derivation on $R.$
%First, consider that for $C=s+tv \in R,$ the equation
%\[
%0=\Delta_\theta^{(C)} \theta(v)+\theta %\Delta_\theta^{(C)}(v)=\Delta_\theta^{(C)}(1-v)+\theta(C)=\theta(C)-C=(s+t-tv)-(s+tv)=%t-2tv,
%\]
%implies that $t=0,$ and hence $C \in \ZZ_4.$  Now, for $a+bv \in R,$
%\begin{itemize}
%\item[(1)]  For $a+bv \in R,$
%\[
%\begin{aligned}
%	\Delta_{\theta} (\theta(a+bv))+\theta (\Delta_{\theta}(a+bv)) &=\Delta_{\theta}(a+b-bv)-\theta((1+2v)(a+b-bv-(a+bv)))\\
%	&=(1+2v)(-b+2bv)+\theta((1+2v)(b-2bv))\\
%	&=0.
%\end{aligned}
%\]
%\item[(2)]  For $a+bv \in R,$
%\[
%\begin{aligned}
%	\Delta_{\theta}(\Delta_{\theta} (a+bv))
%	&=\Delta_{\theta}((1+2v)(a+b-bv-(a+bv)))\\
%	&=\Delta_{\theta}((1+2v)(b-2bv))\\
%	&=0.
%\end{aligned}
%\]

%\item[(3)] Clear from the definition of derivation.
%For $x=a+bv \in R,$ if $\Delta_{\theta}(a+bv)=0$ then  $b=0,$ which implies $\theta(a+bv)=\theta(a)=a=a+bv.$
%\end{itemize}
%\end{proof}

Next, we have

\begin{lem}\label{L-sym}
For all $a \in R,$
we have $x^2a=ax^2.$
\end{lem}

\begin{proof}
Since $xa=\theta(a)x+\Delta_{\theta}(a),$
\[
\begin{aligned}
	x^2a&=x\theta(a)x+x\Delta_{\theta}(a)\\
	    &=(\theta^2(a)x+\Delta_{\theta}(\theta(a)))x+\theta(\Delta_{\theta}(a))x+\Delta_\theta^2(a)\\
	    &=\theta^2(a)x^2+((\Delta_{\theta}\theta+\theta\Delta_{\theta})(a))x+\Delta_{\theta}^2(a)\\
	    &=ax^2 \hspace{1cm}\text{(by Lemma \ref{Property D1})}.
\end{aligned}
\]
\end{proof}

We can generalize Lemma \ref{L-sym} using mathematical induction.

\begin{cor}\label{C-sym}
For all $a \in R,~n \in \ZZ^+,$ we have
\[
x^na=
\begin{cases}
(\theta(a)x+\Delta_\theta(a))x^{n-1},&\text{if } n \text{ is odd},\\
ax^n,&\text{if }n \text{ is even}.
\end{cases}
\]
\end{cor}

\begin{proof}
We know that $xa=\theta(a)x+\Delta_\theta(a)$ and $x^2a=ax^2.$  Suppose the above statement holds for all $n \leq k,$ with $k \geq 2.$  Consider two cases.
If $k+1$ is even, then $k-1$ is also even. From the induction hypothesis,
\[
x^{k+1}a=x^2(x^{k-1}a)=x^2 a x^{k-1}=(ax^2)x^{k-1}=ax^{k+1}.
\]
If $k+1$ is odd, then $k$ is even. Again, from the induction hypothesis,
\[
x^{k+1}a=x(x^k a)=x(a x^k)=(xa)x^k=(\theta(a)x+\Delta_\theta(a))x^k.
\]
Then the result follows.
\end{proof}

% Let us consider the following example.

% \begin{ex}
% 	Let $\Delta_{\theta}(x)=(1+2v)(\theta(x)-x)$ be a derivation on $R.$ Let $f(x)=f_0+f_1x+x^2$ and $g(x)=g_0+g_1x$ be two polynomials in $R[x;\theta,\Delta_{\theta}].$ Then
% 	\[
% 	f(x)+g(x)=(f_0+g_0)+(f_1+g_1)x+x^2=g(x)+f(x).
% 	\]
% 	Also, by applying Corollary \ref{C-sym},
% 	\[
% 	\begin{aligned}
% 		f(x)g(x)&=(f_0+f_1x+x^2)(g_0+g_1x)\\
% 		        &=f_0(g_0+g_1x)+f_1x(g_0+g_1x)+x^2(g_0+g_1x)\\
% 		        &=f_0g_0+f_0g_1x+f_1(\theta(g_0)x+\Delta_{\theta}(g_0))+
% 		        f_1(\theta(g_1)x+\Delta_{\theta}(g_1))x+g_0x^2+g_1x^3\\
% 		        &=(f_0g_0+f_1\Delta_{\theta}(g_0))+(f_0g_1+f_1\theta(g_0)+f_1\Delta_{\theta}(g_1))x+(f_1\theta(g_1)+g_0)x^2+g_1x^3,
% 	\end{aligned}
% 	\]
% 	while
% 	\[
% 	\begin{aligned}
% 		g(x)f(x)&=(g_0+g_1x)(f_0+f_1x+x^2)\\
% 		&=g_0(f_0+f_1x+x^2)+g_1x(f_0+f_1x+x^2)\\
% 		&=g_0f_0+g_0f_1x+g_0x^2+g_1(\theta(f_0)x+\Delta_{\theta}(f_0))+g_1(\theta(f_1)x+\Delta_{\theta}(f_1))x+g_1x^3\\
% 		&=(g_0f_0+g_1\Delta_{\theta}(f_0))+(g_0f_1+g_1\theta(f_0)+g_1\Delta_{\theta}(f_1))x+(g_0+g_1\theta(f_1))x^2+g_1x^3.
% 	\end{aligned}
% 	\]
% 	Therefore, $f(x)g(x) \neq g(x)f(x),$  and hence $R[x;\theta,\Delta_{\theta}]$ is non-commutative.
% 	$\diamondsuit$
% \end{ex}

Let $R^\theta$ be a subset of $R,$ fixed element-wise by $\theta,$ namely $R^\theta:=\{a \in R:~\theta(a)=a\}.$
It is easy to verify that $R^\theta$ is a subring of $R.$  In our case, $R^\theta=\{0,1,2,3\}=\ZZ_4.$   Also, for all $a \in R^\theta,$
 we have $\Delta_{\theta}(a)=0.$  It implies, by Corollary \ref{C-sym}, that for all $a \in R^\theta$ and $n\in \ZZ^+$,
we have $x^n a=ax^n.$

\begin{definition}
	A polynomial $f(x) \in R[x;\theta,\Delta_\theta]$ is called a central element if it satisfies
	\[
	f(x)a(x)=a(x)f(x),
	\]
	for all $a(x) \in R[x;\theta,\Delta_\theta].$  The center of $R[x;\theta,\Delta_{\theta}],$ denoted by $Z(R[x;\theta,\Delta_{\theta}]),$ is defined as
	\[
	Z(R[x;\theta,\Delta_{\theta}]):=\{f(x)\in R[x;\theta,\Delta_{\theta}] :~f(x)a(x)=a(x)f(x),~\text{for all }a(x) \in R[x;\theta,\Delta_{\theta}]\}.
	\]
\end{definition}

%Since $R^\theta=\ZZ_4,$ then for all $a \in R,$ we have $\theta(a)=a$ if and only if $\Delta_\theta(a)=(1+2v)(\theta(a)-a)=0,$ and hence we have proven the following. %\textcolor{red}{bagian ini seperti mengulang Lemma 2.6? Apakah Lemma 2.10 langsung digabung saja dengan Lemma 2.6 atau diubah dalam bahasa R theta?}

%\begin{lem}\label{L-central}
%	Let $a \in R.$  Then $\theta(a)=a$ if and only if $\Delta_\theta(a)=0.$
%\end{lem}

The central element satisfies the following property.

\begin{theorem}\label{T-central}
	$f(x) \in Z(R[x;\theta,\Delta_\theta])$ if and only if $f(x) \in R^\theta[x]$ and for all odd integers $i,$ the coefficient of $x^i$ is equal to $0.$
\end{theorem}

\begin{proof}
	($\Longrightarrow$) Let $f(x)=f_0+f_1x+f_2x^2+\cdots+f_kx^k.$  Observe that
	\[
	xf(x)=\sum_{i=0}^k (xf_i)x^i=\sum_{i=0}^k (\theta(f_i)x+\Delta_\theta(f_i))x^i=\Delta_{\theta}(f_0)
	+\sum_{i=1}^{k}\left(\theta(f_{i-1})+\Delta_{\theta}(f_i)\right)x^i
	+\theta(f_k)x^{k+1}
	\]
	and
	\[
	f(x)x=f_0x+f_1x^2+\cdots+f_kx^{k+1}.
	\]
	Since $f(x)$ is a central element, $xf(x)=f(x)x,$ which implies
	\begin{equation}\label{P1}
		\Delta_\theta(f_0)=0,
	\end{equation}
	\begin{equation}\label{P2}
		\theta(f_{i-1})+\Delta_\theta(f_i)=f_{i-1}, \text{ for }1 \leq i \leq k,
	\end{equation}
	\begin{equation}\label{P3}
		\theta(f_k)=f_k.
	\end{equation}

The Equation (\ref{P3}) implies $f_k \in R^\theta.$  By Lemma \ref{Property D1} we have $\Delta_{\theta}(f_k)=0.$  By substituting the Equation (\ref{P2}) repeatedly we obtain
$f_1,f_2,\ldots,f_{k-1} \in R^\theta.$  Moreover, by Equation (\ref{P1}) and Lemma \ref{Property D1} obtain $f_0 \in R^\theta.$  Thus, $f(x) \in R^\theta[x].$

Now, take $a=v \in R.$  Observe that
\[
af(x)=af_0+af_1x+ \cdots+af_k x^k
\text{ and }
f(x)a=f_0a + \sum_{i=1}^{k} f_i(x^ia).
\]
%and
%\[
%f(x)a=f_0a + \sum_{i=1}^{k} f_i(x^ia).
%\]
In this case, for $0 \leq 2j < k,$ the coefficient of $x^{2j}$ in $af(x)$ and $f(x)a$ is equal to $af_{2j}$ and $f_{2j}a+f_{2j+1}\Delta_{\theta}(a),$ respectively.  Since $f(x)$ is a central element and $\Delta_{\theta}(a)=1$, we conclude that $f_{2j+1}=0.$

\noindent
$(\Longleftarrow)$  Let $f(x)=f_0+f_2x^2+f_4x^4+\cdots+f_{2m}x^{2m} \in R^\theta[x].$  Let $a(x)=a_0+a_1x+a_2x^2+\cdots+a_kx^k \in R[x;\theta,\Delta_{\theta}].$  Observe that
\[
(f_{2i}x^{2i})(a_jx^j)=f_{2i}(x^{2i}a_j)x^j=f_{2i}a_jx^{2i+j},
\]
and
\[
\begin{aligned}
	(a_jx^j)(f_{2i}x^{2i})&= a_j(x^jf_{2i})x^{2i}\\
	&=\begin{cases}
		a_j\theta(f_{2i})x^{2i+j}+a_j\Delta_{\theta}(f_{2i})x^{2i+j-1}, &\text{if } j \text{ is odd},\\
		a_jf_{2i}x^{2i+j}, &\text{if } j \text{ is even,}
	\end{cases}
	\\
	&=a_jf_{2i}x^{2i+j} \text{ (by Lemma \ref{Property D1}).}
\end{aligned}
\]
Then $(f_{2i}x^{2i})(a_jx^j)=(a_jx^j)(f_{2i} x^{2i}),$ for all $i,j.$  Hence, $f(x)a(x)=a(x)f(x),$ for all $a(x) \in R[x;\theta,\Delta_{\theta}].$
\end{proof}

%By using the same way with the one in \cite{Sharma2018}, we also have the %right division algorithm as follows.  We provide the proof here for the %readers' convenience.

We end this section by establishing the right-division algorithm below.

\begin{lem}[Right-Division Algorithm]
	Let $f(x),g(x) \in R[x;\theta,\Delta_{\theta}]$ such that the leading coefficient of $g(x)$ is a unit.  Then there exist $q(x), r(x) \in R[x;\theta,\Delta_{\theta}]$ such that
	\[
	f(x)=q(x)g(x)+r(x),
	\]
	with $r(x)=0$ or $\deg r(x) < \deg g(x).$
\end{lem}

\begin{proof} Similar to the proof of Theorem 2.8 in \cite{Sharma2018} and Theorem 1 in \cite{Patel2022}.
%	Let $f(x)=f_0+f_1x+f_2x^2+\cdots+f_rx^r$ and $g(x)=g_0+g_1x+g_2x^2+\cdots+g_sx^s,$ where $g_s$ is unit, be two polynomials in $R[x;\theta,\Delta_{\theta}].$  If $r<s$, we can choose $q(x)=0$ and $r(x)=f(x)$ to get the desired result. If $r \geq s,$ define a polynomial $h(x):=f(x)-A(x)g(x),$ where
%	\[
%	A(x)=\begin{cases}
%		f_r \theta(g_s^{-1})x^{r-s}, & r-s \text{ is odd},\\
%		f_rg_s^{-1} x^{r-s}, & r-s \text{ is even}.
%	\end{cases}
%	\]
%	If $r-s$ is odd, then the coefficient of $x^r$ in $A(x)g(x)$ is equal to the coefficient of $x^r$ in $f_r \theta(g_s^{-1})x^{r-s}g_sx^s,$ namely
%	\[
%	f_r \theta(g_s^{-1}) \theta(g_s)=f_r \theta(g_s^{-1} g_s)=f_r.
%	\]
%	If $r-s$ is even, then the coefficient of $x^r$ in $A(x)g(x)$ is equal to the coefficient of $x^r$ in $f_r g_s^{-1}x^{r-s}g_sx^s,$ namely
%	\[
%	f_r g_s^{-1} g_s=f_r.
%	\]
%	Thus, the coefficient of $x^r$ in $f(x)-A(x)g(x)$ is equal to $0,$ and hence the degree of $h(x)$ is less than $r$. We prove the result by induction on $r$. The case when $r=0$ is obvious. Suppose that the result holds for every polynomial having degree less than $r>0$. By the induction hypothesis, we obtain
%	\[
%	h(x)=q_1(x)g(x)+r_1(x),
%	\]
%	with $r_1(x)=0$ or $\deg(r_1(x)) < \deg(g(x)).$
%	Therefore,
%	\[
%	f(x)=h(x)+A(x)g(x)=(q_1(x)+A(x))g(x)+r_1(x)=q(x)g(x)+r(x),
%	\]
%	with $q(x)=q_1(x)+A(x)$ and $r(x)=r_1(x).$  Hence, the result follows.
\end{proof}
%\textcolor{red}{apakah akan ada contoh seperti paper konferensi?}
%Let us consider an example below.

%\begin{ex}
%	\color{red} Tolong tambahkan contoh di sini ya \normalcolor
%	$\diamondsuit$
%\end{ex}
\section{$\Delta_{\theta}$-cyclic codes over $R$}

For $f(x)$ a polynomial of degree $n$ in $R[x;\theta,\Delta_{\theta}],$ let
\[
\langle f(x) \rangle=\{a(x)f(x):~a(x) \in R[x;\theta,\Delta_{\theta}]\}.
\]
It is easy to see that $R[x;\theta,\Delta_{\theta}]/\langle f(x) \rangle$ is a left module over $R[x;\theta,\Delta_{\theta}],$ where the scalar multiplication is defined by
\[
r(x) \left( a(x)+\langle f(x)\rangle \right):=r(x)a(x)+\langle f(x)\rangle.
\]

\begin{definition}
	A code $C \seq R^n$ is called a $\Delta_{\theta}$-linear code of length $n$ over $R$ if $C$ is a left $R[x;\theta,\Delta_{\theta}]$-submodule of $R[x;\theta,\Delta_{\theta}]/\langle f(x) \rangle$
	for $f(x) \in R[x;\theta,\Delta_{\theta}]$ a polynomial of degree $n.$ If $f(x)$ is a central element, then $C$ is called a central $\Delta_{\theta}$-linear code.
\end{definition}

\begin{definition}
	A code $C \seq R^n$ is called a $\Delta_{\theta}$-cyclic code of length $n$ over $R$ if $C$ is a $\Delta_{\theta}$-linear code and for all $\cv=(c_0,c_1,\ldots,c_{n-1}) \in C$ we have
	\[
	T_{\Delta_{\theta}}(\cv):=(\theta(c_{n-1})+\Delta_{\theta}(c_0),\theta(c_0)+\Delta_{\theta}(c_1),\ldots,\theta(c_{n-2})+\Delta_{\theta}(c_{n-1})) \in C.
	\]
	Here, $T_{\Delta_{\theta}}$ is called a $\Delta_{\theta}$-cyclic shift operator.
\end{definition}

\begin{rem}
	If $\theta$ is the identity automorphism and $\Delta_{\theta}\equiv0,$ then we obtain a (usual) cyclic code.  Hence, the $\Delta_{\theta}$-cyclic code is a generalization of a cyclic code.  Moreover, recall that if $\theta$ is the identity automorphism and $\Delta_{\theta}\equiv 0,$ then a linear code $C \seq R^n$ is called quasi-cyclic of index $k$ ($k$ is a divisor of $n$) if for all $\cv \in C,$ we have $T^k_{\Delta_{\theta}}(\cv) \in C.$ Here $T^k_{\Delta_{\theta}}(\cv)=\underbrace{(T_{\Delta_{\theta}} \circ T_{\Delta_{\theta}} \circ \cdots \circ T_{\Delta_{\theta}})}_{k}(\cv),$ the composition of $k$ numbers of $\Delta_{\theta}$-cyclic shift operator.
\end{rem}

For our purpose, to convert the algebraic structures of $\Delta_{\theta}$-cyclic codes into combinatorial structures and vice versa, we consider the following correspondence:
\[
\begin{array}{ccc}
	R[x;\theta,\Delta_{\theta}]/\langle f(x) \rangle & \longrightarrow & R^n,\\
      c_0+c_1x+c_2x^2+\cdots+c_{n-1}x^{n-1} & \longmapsto  &
      (c_0,c_1,\ldots,c_{n-1}).   	
\end{array}
\]

From now on, let $R_{n,\Delta_{\theta}}:=R[x;\theta,\Delta_{\theta}]/\langle x^n-1\rangle.$

\begin{lem}\label{L-xc(x)identified}
	If $c(x)=c_0+c_1x+c_2x^2+\cdots+c_{n-1}x^{n-1} \in R_{n,\Delta_{\theta}}$ %$\color{red}:=R[x;\theta,\Delta_{\theta}]/\langle x^n-1\rangle.$ \textcolor{red}{di definisi %awal f(x) sembarang, tapi kalau skew-cyclic f(x) harus }$\color{red}x^n-1?$  %R[x;\theta,\Delta_{\theta}]/\langle x^2-1\rangle$
	is identified by a codeword $\cv=(c_0,c_1,\ldots,c_{n-1}) \in R^n,$ then $xc(x) \in R_{n,\Delta_{\theta}}$ is identified by $T_{\Delta_{\theta}}(\cv) \in R^n.$
\end{lem}

\begin{proof}
	\[
	\begin{aligned}
		xc(x) &=x \left( \sum_{i=0}^{n-1} c_ix^i \right)\\
		      %&= \sum_{i=0}^{n-1}(x c_i)x^i\\
%		      &= \sum_{i=0}^{n-1} (\theta(c_i)x+\Delta_{\theta}(c_i))x^i\\
%		      &= \sum_{i=0}^{n-1} \theta(c_i)x^{i+1}+\sum_{i=0}^{n-1}\Delta_{\theta}(c_i)x^i\\
		      &= \sum_{i=1}^{n-1} \left(\theta(c_{i-1})+\Delta_{\theta}(c_i) \right) x^i + (\theta(c_{-1})+\Delta_{\theta}(c_0)) ~(\text{denoting that } c_{-1}=c_{n-1})\\
		      &=\sum_{i=0}^{n-1} (\theta(c_{i-1})+\Delta_{\theta}(c_i))x^i.
	\end{aligned}
	\]
	It means that $xc(x)$ is identified by $T_{\Delta_{\theta}}(\cv) \in R^n.$
\end{proof}

\begin{lem}\label{L-c-submodule}
	A code $C \seq R^n$ is $\Delta_{\theta}$-cyclic code if and only if $C$ is a left $R[x;\theta,\Delta_{\theta}]$-submodule of $R_{n,\Delta_{\theta}}.$
	%$\color{red}:=R[x;\theta,\Delta_{\theta}]/\langle x^n-1\rangle.$
\end{lem}

\begin{proof}
	$(\Longrightarrow)$  Since $C$ is a $\Delta_{\theta}$-linear code, $(C,+)$ is a subgroup of $(R_{n,\Delta_{\theta}},+),$  and for any $\cv \in C,$ identified by $c(x) \in R_{n,\Delta_{\theta}},$ we have $T_{\Delta_{\theta}}(\cv) \in C.$ By Lemma \ref{L-xc(x)identified}, $T_{\Delta_{\theta}}(\cv) \in C$ can be identified by $xc(x) \in R_{n,\Delta_{\theta}}.$
	%R[x;\theta,\Delta_{\theta}]/\langle x^n-1\rangle.$
	Inductively, we obtain $x^i c(x) \in C,$ for all $i \in \ZZ^+.$  By the linearity of the scalar multiplication, we have $a(x) c(x) \in C$ for all $a(x) \in R[x; \theta, \Delta_{\theta}].$ Hence, $C$ is a left submodule of $R_{n,\Delta_{\theta}}.$
	
	\noindent
	$(\Longleftarrow)$  If $C$ is a left submodule of $R_{n,\Delta_{\theta}},$ then for all $\cv \in C,$ identified by $c(x) \in R_{n,\Delta_{\theta}},$
	%R[x;\theta,\Delta_{\theta}]/ \langle x^n-1 \rangle,$
	we have $xc(x) \in C.$  By Lemma \ref{L-xc(x)identified}, $xc(x)$ can be identified by $T_{\Delta_{\theta}}(\cv) \in R^n.$  Hence, $T_{\Delta_{\theta}}(\cv) \in C,$ which implies that $C$ is a $\Delta_\theta$-cyclic code.
\end{proof}

\begin{cor}
	If $C \seq R^n$ is a $\Delta_{\theta}$-cyclic code of even length $n,$ then $C$ is an ideal of $R_{n,\Delta_{\theta}}.$ %$\color{red}:=R[x;\theta,\Delta_{\theta}]/\langle x^n-1\rangle.$
\end{cor}

\begin{proof}
	Since $n$ is even, by Theorem \ref{T-central}, $x^n-1$ is a central element.  Then for all $a(x) \in R[x;\theta,\Delta_{\theta}],$ we have $a(x)(x^n-1)=(x^n-1)a(x).$  Then $\langle x^n-1\rangle$ is a two-sided ideal of $R[x;\theta,\Delta_{\theta}].$ Since $C$ is a left submodule of $R_{n,\Delta_{\theta}},$ $C$ is an ideal of $R_{n,\Delta_{\theta}}.$
\end{proof}

\begin{lem}\label{L-c-qc}
	Let $C \seq R^n$ be a $\Delta_{\theta}$-cyclic code.  Then the following two statements hold.
	\begin{itemize}
		\item[(1)] $C$ is a cyclic code of length $n$ over $R$ if $n$ is odd;
		\item[(2)] $C$ is a quasi-cyclic code of length $n$ and index $2$ over $R$ if $n$ is even.
	\end{itemize}
	\end{lem}

\begin{proof}
	\begin{itemize}
		\item[(1)]  If $n$ is odd, then there exists $b \in \ZZ$ such that $2b=n+1.$  Let $\cv=(c_0,c_1,\ldots,c_{n-1}) \in C$ be a codeword identified by $c(x)=c_0+c_1x+c_2x^2+\cdots+c_{n-1}x^{n-1} \in R_{n,\Delta_{\theta}}.$
		%R[x;\theta,\Delta_\theta]/ \langle x^n-1 \rangle.$
		It is clear that $x^{2b}c(x) \in C.$  Observe that
		\[
		x^{2b}c(x)=x^{2b}\sum_{i=0}^{n-1} c_ix^i=\sum_{i=0}^{n-1} (x^{2b} c_i)x^i=\sum_{i=0}^{n-1} c_ix^{n+1+i}=\sum_{i=0}^{n-1} c_{i-1}x^i ~(\text{denoting that } c_{-1}=c_{n-1}).
		\]
		Note that $x^{2b}c_i=c_ix^{2b}$ is derived from Corollary \ref{C-sym}.  The equation above says that the codeword $x^{2b}c(x) \in C$ can be identified by the vector $(c_{n-1},c_0,c_1,\ldots,c_{n-2})\in C.$  Thus, $C$ is a cyclic code.
		
		\item[(2)]  Observe that for every vector $\cv=(c_0,c_1,\ldots,c_{n-1})\in C$ identified by the polynomial $c(x)=c_0+c_1x+c_2x^2+\cdots+c_{n-1}x^{n-1} \in
		R_{n,\Delta_{\theta}}$
		%R[x;\theta,\Delta_{\theta}] / \langle x^n-1 \rangle,$
		we have $x^2c(x) \in C.$  By the similar way as in part (1), we can show that $x^2c(x)$ can be identified by $(c_{n-2},c_{n-1},c_0,\ldots,c_{n-3}) \in C.$  Then $C$ is a quasi-cyclic code of index $2.$
	\end{itemize}
\end{proof}

\begin{lem}
	If $C \seq R^n$ is a $\Delta_{\theta}$-cyclic code and $g(x)$ is a nonzero polynomial in $C$ of smallest degree with leading coefficient is a unit, then the following three statements hold.
	\begin{itemize}
		\item[(1)]  $C=\langle g(x) \rangle.$
		
		\item[(2)]  $g(x)$ is a right divisor of $x^n-1.$
		
		\item[(3)]  $\{g(x),xg(x),\ldots,x^{n-k-1}g(x)\}$ is a basis of $C,$ with $k=\deg(g(x)).$
	\end{itemize}\label{L-3cond}
\end{lem}

\begin{proof}
	\begin{itemize}
		\item[(1)]  It is clear that $\langle g(x) \rangle \seq C.$ Let $c(x) \in C.$  By the right-division algorithm, there exist $q(x),r(x) \in R[x;\theta,\Delta_{\theta}]$ such that
		\[
		c(x)=q(x)g(x)+r(x),
		\]
		with $r(x)=0$ or $\deg(r(x))<\deg(g(x)).$  Since $C$ is a left submodule of $R_{n,\Delta_{\theta}}$ over $R[x;\theta,\Delta_{\theta}],$ $r(x)=c(x)-q(x)g(x) \in C.$  Moreover, since $g(x)$ is a nonzero polynomial of smallest degree, $r(x)=0.$  Thus, $c(x)=q(x)g(x) \in \langle g(x)\rangle,$ and hence $C \seq \langle g(x)\rangle.$
		
		\item[(2)]  Again, by right-division algorithm, there exist $q(x),r(x) \in R[x;\theta,\Delta_{\theta}]$ such that
		\[
		x^n-1=q(x)g(x)+r(x),
		\]
		with $r(x)=0$ or $\deg(r(x))<\deg(g(x)).$  By the same argument with above, we have $x^n-1=q(x)g(x).$ It means that $g(x)$ is is a right divisor of $x^n-1.$
		
		\item[(3)]  From the above result, let $x^n-1=h(x)g(x),$ for some $h(x) \in R[x;\theta,\Delta_{\theta}].$  Let $c(x) \in C=\langle g(x) \rangle$ and let $c(x)= \ell(x) g(x)$ for some $\ell(x) \in R[x;\theta,\Delta_{\theta}].$  If $\deg(\ell(x)) \leq n-k-1,$ then $c(x) \in \langle g(x),xg(x),\ldots,x^{n-k-1}g(x) \rangle.$  If $\deg(\ell(x)) > n-k-1,$
		by using the right-division algorithm we come up with the conclusion that $c(x)=\ell(x)g(x)=r(x)g(x)$ (in $R_{n,\Delta_{\theta}}$) for some polynomial $r(x)\in R[x;\theta,\Delta_{\theta}]$ with $r(x)=0$ or $\deg(r(x))<\deg(h(x))=n-k$, so $\{g(x),xg(x),\ldots,x^{n-k-1}g(x)\}$  generates $C.$  Moreover, it is easy to see that the set is also linearly independent.  Hence, we conclude that  $\{g(x),xg(x),\ldots,x^{n-k-1}g(x)\}$ is a basis for $C.$
	\end{itemize}
\end{proof}

Since $\langle g(x) \rangle$ is a submodule of $R_{n,\Delta_{\theta}}$ over $R[x;\theta,\Delta_{\theta}],$ by Lemma \ref{L-c-submodule}, $C$ is a $\Delta_{\theta}$-cyclic code.  The Lemma \ref{L-3cond} above says that if $C=\langle g(x) \rangle$ is a $\Delta_{\theta}$-cyclic code, where $g(x)$ is a nonzero polynomial in $C$ of smallest degree with a unit leading coefficient and $g(x)$ is also a right divisor of $x^n-1,$ then $C$ is free. Hence, we obtained a construction method for $\Delta_{\theta}$-cyclic codes as described by the following.

\begin{cor}\label{Construction}
	If $g(x)$ is a right divisor of $x^n-1,$ with leading coefficient a unit, then $C=\langle g(x)\rangle$ is a free $\Delta_{\theta}$-cyclic code.
\end{cor}

%If $C$ is a free $\Delta_{\theta}$-cyclic code, then the existence of a generating polynomial of $C$ is guaranteed.

%\begin{lem}
%	If $C\seq R^n$ is a free $\Delta_{\theta}$-cyclic code, then there exists a nonzero polynomial of smallest degree $g(x) \in C$ with leading coefficient a unit such that $C= \langle g(x) \rangle$ where $g(x)$ is a right divisor of $x^n-1.$
%\end{lem}

%\begin{proof}
%	Let $g(x) \in C$ is a nonzero polynomial of smallest degree.  Suppose the leading coefficient of $g(x)$ is not a unit, then it should be a zero divisor (since $R$ is a finite ring).  Let $a$ be a leading coefficient of $g(x).$  Since $a$ is a zero divisor, then there exists $r \in R,~r \neq 0,$ such that $ra=0.$  Since $C$ is a linear code, then $rg(x) \in C.$ It is clear that $rg(x)$ is either a nonzero polynomial of degree less than the degree of $g(x)$ or $rg(x)=0.$  But the last equation is impossible since it implies $\{g(x)\}$ is linearly dependent.  Hence, the leading coefficient of $g(x)$ should be a unit.  Thus, by Lemma \ref{L-3cond}, we have $C=\langle g(x) \rangle,$ with $g(x)$ is a right divisor of $x^n-1.$
%\end{proof}

%Since $\langle g(x) \rangle$ is a submodule of $R_{n,\Delta_{\theta}}$ over $R[x;\theta,\Delta_{\theta}],$ then by Lemma \ref{L-c-submodule}, $C$ is a $\Delta_{\theta}$-cyclic code.  We obtained a construction method for $\Delta_{\theta}$-cyclic codes as described by the following.

%\begin{cor}
%	If $g(x)$ is a right divisor of $x^n-1,$ with leading coefficient a unit, then $C=\langle g(x)\rangle$ is a free $\Delta_{\theta}$-cyclic code.
%\end{cor}

 Let $C= \langle g(x) \rangle$ be a free  $\Delta_{\theta}$-cyclic code of length $n$ generated by a right divisor $g(x)$ of $x^n-1,$ whose leading coefficient is a unit.  Then the generator matrix $G$ of $C$ of dimension $(n-k) \times n$ is given by
\[
G=\begin{pmatrix}
	g(x)\\xg(x)\\x^2 g(x)\\ \vdots\\x^{n-k-1} g(x)
\end{pmatrix},
\]
where $g(x)=g_0+g_1x+g_2x^2+\cdots+g_kx^k.$  To be more precise, if $n-k$ is odd, then
\[
G=\left(
\begin{smallmatrix}
	g_0 & g_1 & g_2 & \cdots & g_k & 0 & \cdots & 0\\
	\Delta_{\theta}(g_0) & \theta(g_0)+\Delta_{\theta}(g_1) & \theta(g_1)+\Delta_{\theta}(g_2) & \cdots & \theta(g_{k-1})+\Delta_{\theta}(g_k) & \theta(g_k) & \cdots & 0\\
	0 & 0 & g_0 & \cdots & g_{k-3} & g_{k-2} & \cdots & 0\\
	\vdots & \vdots & \vdots & \cdots & \vdots &  &  & \vdots\\
	0 & 0 & \cdots & g_0 & g_{k-2} & \cdots & g_{k-1} & g_k
\end{smallmatrix}
\right),
\]

and if $n-k$ is even, then

%\begin{multline*}
\[
G =\left(
\begin{smallmatrix}
	g_0 & g_1 & g_2 & \cdots & g_k & 0 & \cdots & 0 & 0\\
	\Delta_{\theta}(g_0) & \theta(g_0)+\Delta_{\theta}(g_1) & \theta(g_1)+\Delta_{\theta}(g_2) & \cdots & \theta(g_{k-1})+\Delta_{\theta}(g_k) & \theta(g_k) & \cdots & 0 & 0\\
	0 & 0 & g_0 & \cdots & g_{k-3} & g_{k-2} & \cdots & 0 & 0\\
	\vdots & \vdots & \vdots & \cdots & \vdots & \vdots & \vdots & \vdots &  \vdots\\
	0 & 0 & \cdots & \Delta_{\theta}(g_0) & \theta(g_0)+\Delta_{\theta}(g_1) & \theta(g_1)+\Delta_{\theta}(g_2) & \cdots & \theta(g_{k-1})+\Delta_{\theta}(g_k) & \theta(g_k)\\
	\end{smallmatrix}
\right).
\]
%\right.
%\\
%\left.
%\begin{smallmatrix}
%0 & \cdots & 0 & 0\\
%\theta(g_k) & \cdots & 0 & 0\\
%g_{k-2} & \cdots & 0 & 0\\
%\vdots & \vdots & \vdots &  \vdots\\
%\theta(g_1)+\Delta_{\theta}(g_2) & \cdots & %\theta(g_{k-1})+\Delta_{\theta}(g_k) & \theta(g_k)
%\end{smallmatrix}
%\right).
%\end{multline*}
%\textcolor{red}{font size matriksnya perlu dikecilkan?}

%\textcolor{red}{apakah perlu contoh?}

%begin{ex}
%	Tambahkan contoh di sini \qed
%\end{ex}

\section{Dual of $\Delta_{\theta}$-cyclic codes over $R$}

In this section, we investigate the structure of the dual of a free $\Delta_{\theta}$-cyclic code over $R$ with even length.

%\subsection{Dual and parity-check matrix}

\begin{definition}
	Let $C \seq R^n$ is a $\Delta_{\theta}$-cyclic code.  Dual of $C,$ denoted by $C^\perp,$ is defined by
	\[
	C^\perp=\{\xv \in R^n:~\xv \cdot \yv=0, \text{ for all }\yv \in C\}.
	\]
	Here, $\xv \cdot \yv=x_0y_0+x_1y_1+\cdots+x_{n-1}y_{n-1}$ denotes the usual inner product in $R^n.$
\end{definition}

It is easy to check that $C^\perp$ is a linear code over $R.$  Moreover, for any even $n,$ if $C$ is a free $\Delta_{\theta}$-cyclic code principally generated by a polynomial whose leading coefficient is a unit, then $C^\perp$ is free, as we show below.

\begin{lem}
	Let $n$ be an even integer. Let $g(x),h(x) \in R[x;\theta,\Delta_{\theta}],$ where the leading coefficient of $g(x)$ is a unit, and $h(x)g(x)=x^n-1.$ Then we have
	\[
	h(x)g(x)=g(x)h(x).
	\]
\end{lem}

\begin{proof}
    %\textcolor{red}{(bukti sedikit direvisi)}
	If $n$ is even, then $x^n-1=h(x)g(x)$ is a central element (by Theorem \ref{T-central}).  Then we have
	\[
	h(x)h(x)g(x)=h(x)g(x)h(x).
	\]
	Since $h(x)$ is not a zero divisor, $x^n-1=h(x)g(x)=g(x)h(x).$
\end{proof}

%Using this commutativity, we obtain a parity check polynomial.

\begin{lem} \label{L-ch=0}
	Let $C$ be a $\Delta_{\theta}$-cyclic code, where $C=\langle g(x)\rangle,$ for some right divisor $g(x)$ of $x^n-1,$ whose leading coefficient is a unit and $n$ is even.  Let $x^n-1=h(x)g(x).$  Then $c(x) \in R_{n,\Delta_{\theta}}$ is contained in $C$ if and only if $c(x)h(x)=0$ (in $R_{n,\Delta_{\theta}}$).
\end{lem}

\begin{proof}
	$(\Longrightarrow)$  Since $c(x) \in C,$ there exists $a(x) \in R[x;\theta,\Delta_{\theta}]$ such that $c(x)=a(x) g(x).$  Hence,
	\[
	c(x) h(x)=a(x)g(x)h(x)=a(x)h(x)g(x)=a(x)(x^n-1)=0
	\text{ (in }R_{n,\Delta_{\theta}}\text{).}\]
	
	$(\Longleftarrow)$  Since $c(x)h(x)=0$ (in $R_{n,\Delta_{\theta}}$) for some $c(x) \in R_{n,\Delta_{\theta}},$ we have $c(x)h(x)=q(x)(x^n-1)$ for some $q(x) \in R[x;\theta,\Delta_{\theta}].$  In this case
	\[
	c(x)h(x)=q(x)(x^n-1)=q(x)h(x)g(x)=q(x)g(x)h(x).
	\]	
	Since $h(x)$ is not a zero divisor, we have $c(x)=q(x)g(x) \in \langle g(x) \rangle=C.$

\end{proof}

We also have known that all unit in $R$ are $\{1,3,1+2v,3+2v\}.$  It is easy to check that for all units $a$ in $R,$ we have $\theta(a)'$s are also units in $R.$  Hence, we have the following.

\begin{lem}\label{L-unit}
	If $a \in R$ is a unit, then $\theta(a) \in R$ is also a unit.
\end{lem}

Now, consider a $\Delta_{\theta}$-cyclic code $C$ of even length $n.$ Let
$C=\langle g(x)\rangle,$  where $g(x)$ is a right divisor of $x^n-1$ and its leading coefficient is a unit.  Then there exists $h(x)=h_0+h_1x+\cdots+h_kx^k \in R_{n,\Delta_{\theta}}$ such that $x^n-1=h(x)g(x).$  For $c(x)=c_0+c_1x+\cdots+c_{n-1}x^{n-1} \in C,$ by Lemma \ref{L-ch=0}, we have
$c(x)h(x)=0$ (in $R_{n,\Delta_{\theta}}$).  By considering the coefficients of $x^k,x^{k+1},\ldots,x^{n-1}$ in the last equation, we obtain two systems of equations each consisting of $n-k$ linear equations in $c_0,c_1,\ldots,c_{n-1}.$  To be more precise, we have the following two systems of equations, for $k$ is odd and even, respectively.
\[
\begin{cases}
	c_ih_k+c_{i+1}(\theta(h_{k-1})+\Delta_{\theta}(h_k))+c_{i+2}h_{k-2}+\cdots+c_{k+i}(\theta(h_0)+\Delta_{\theta}(h_1))=0, & \text{if }i \text{ is even,}\\
	
	c_i\theta(h_k)+c_{i+1}h_{k-1}+c_{i+2}(\theta(h_{k-2})+\Delta_{\theta}(h_{k-1}))+\cdots+c_{k+i}h_0+c_{k+i+1}\Delta_{\theta}(h_0)=0, & \text{if }i \text{ is odd.}
\end{cases}
\]
\[
\begin{cases}
	c_ih_k+c_{i+1}(\theta(h_{k-1})+\Delta_{\theta}(h_k))+c_{i+2}h_{k-2}+\cdots+c_{k+i}h_0+c_{k+i+1}\Delta_{\theta}(h_0)=0, & \text{if }i\text{ is even,}\\
	
	c_i\theta(h_k)+c_{i+1}h_{k-1}+c_{i+2}(\theta(h_{k-2})+\Delta_{\theta}(h_{k-1}))+\cdots+c_{k+i}(\theta(h_0)+\Delta_{\theta}(h_1))=0, & \text{if }i \text{ is odd,}
\end{cases}
\]
If we write the system of equations in a matrix form, we obtain $H\cv^T=0,$ for the matrix $H$ of dimension $(n-k) \times n.$  It implies that $GH^T=0,$ for a generator matrix $G$ of $C.$ Moreover, it is easy to check that $H$ is of the row echelon form with diagonal elements $h_k$ or $\theta(h_k),$ and having a submatrix of dimension $(n-k) \times (n-k).$  Since $h_k$ is a unit, by Lemma \ref{L-unit}, we have $\theta(h_k)$ is also a unit.  Then the rows of $H$ are linearly independent.  Hence,
\[
|\text{Row space of }H|=|R|^{n-k}=|C^\perp|.
\]
Thus, $H$ is a parity check matrix of $C$ and we have proven the following theorem.

\begin{theorem}
	Let $C$ be a $\Delta_{\theta}$-cyclic code of even length $n$, where $C=\langle g(x)\rangle,$ for some right divisor $g(x)$ of $x^n-1,$ whose leading coefficient is a unit and $x^n-1=h(x)g(x),$ for some $h(x)=h_0+h_1x+h_2x^2+\cdots+h_kx^k \in R_{n,\Delta_{\theta}}.$  Then the $(n-k) \times n$ parity-check matrix $H$ for $C$ is given by
	\[
	\left(
	\begin{smallmatrix}
		h_k & \theta(h_{k-1})+\Delta_{\theta}(h_k) & h_{k-2} & \cdots & \cdots & \cdots & \theta(h_0)+\Delta_{\theta}(h_1) & \cdots & \cdots & \cdots & \cdots & \cdots & 0 & 0\\
		0 & \theta(h_k) & h_{k-1} & \theta(h_{k-2}) & \cdots &\cdots  & \cdots &\Delta_{\theta}(h_0) & \cdots & \cdots & \cdots & \cdots & 0 & 0\\
		0 & 0 & h_k & \theta(h_{k-1})+\Delta_{\theta}(h_k) & h_{k-2} & \cdots & \cdots & \cdots & \theta(h_0)+\Delta_{\theta}(h_1) & \cdots & \cdots & \cdots & 0 & 0\\
		\vdots & \vdots & \vdots & \vdots & \vdots & \ddots & \ddots & \ddots & \ddots & \ddots & \ddots & \ddots & \vdots & \vdots\\
		0 & 0 & \cdots & \cdots & 0 & h_k & \theta(h_{k-1})+\Delta_{\theta}(h_k) & h_{k-2} &
		\cdots & \cdots & \cdots & \cdots & h_1 & \theta(h_0)+\Delta_{\theta}(h_1)
	\end{smallmatrix}
	\right)
	\]
	for an odd $k,$ and
	\[
	\left(
	\begin{smallmatrix}
		h_k & \theta(h_{k-1})+\Delta_{\theta}(h_k) & h_{k-2} & \cdots & \cdots & \cdots & h_0 & \Delta_{\theta}(h_0) & 0 & \cdots & \cdots & \cdots  & 0 & 0\\
		0 & \theta(h_k) & h_{k-1} & \theta(h_{k-2}) & \cdots &\cdots & h_1 &\theta(h_0)+\Delta_{\theta}(h_1) & 0 & \cdots  & \cdots & \cdots & 0 & 0\\
		0 & 0 & h_k & \theta(h_{k-1})+\Delta_{\theta}(h_k) & h_{k-2} & \cdots & h_2 &\theta(h_1)+\Delta_{\theta}(h_2) & \cdots & \cdots & \cdots & \cdots & 0 & 0\\
		\vdots & \vdots & \vdots & \vdots & \vdots & \ddots & \ddots & \ddots & \ddots & \ddots & \ddots  & \vdots & \vdots & \vdots\\
		0 & 0 & \cdots & \cdots & 0 & \theta(h_k) & h_{k-1} & \cdots &
		\cdots & \cdots & \cdots & \cdots & h_1 & \theta(h_0)+\Delta_{\theta}(h_1)
	\end{smallmatrix}
	\right)
	\]
	for an even $k.$
\end{theorem}
	
\section{New linear codes over $\ZZ_4$}\label{Application}

In this section, we obtained many linear codes over $\ZZ_4$ with new parameters from the Gray image, residue code, and torsion code of skew-linear and skew-cyclic codes with derivation over $R$. First, we construct linear codes over $\ZZ_4$ from principally generated free $\Delta_\theta$-cyclic codes over $R$ in Corollary \ref{Construction}. The linear codes over $\ZZ_4,$ some of them are with new parameters, are listed in the table below.
%\newpage
\begin{center}
\begin{longtable}{c|c|c|c|c}
\hline \hline
{\multirow{2}{*}{$C$}} & $\Phi(C)$ & $Res(C)$ & $Tor(C)$ & $C^{PS}$ \\
\cline{2-5}
& $[n,4^{k_1}2^{k_2},d_L]$ & $[n,4^{k_1}2^{k_2},d_L]$ & $[n,4^{k_1}2^{k_2},d_L]$ & $[n,4^{k_1}2^{k_2},d_L]$ \\
\hline
$\langle 3+x \rangle$ & $[8,4^6 2^0,2]^{\ast)}$ & $[4,4^3 2^0,2]^{\ast)}$&$[4,4^3 2^0,2]^{\ast)}$& $[8,4^6 2^0,2]^{\ast)}$\\
\hline
$\begin{aligned}\langle (3+2v)+(3+2v)x+2x^2\\
	+(1+2v)x^3+(3+2v)x^4 \rangle
\end{aligned}$ &&& $[6,4^2 2^0,6]^{\ast)}$ &\\
\hline

$\begin{aligned}\langle 3v+(3+v)x+(3+v)x^2+(1+2v)x^3\\
+(2+2v)x^4+2x^5+vx^6+(1+3v)x^7\\
	+(1+v)x^8+x^9 \rangle
\end{aligned}$ &&$[12,4^3 2^0,10]^{\ast)}$&&\\
\hline
$\begin{aligned}\langle (1+3v)+3x+x^2+(3+2v)x^4\\
+2x^5+2vx^6+2x^7+(1+3v)x^8\\
	+3x^9+3x^{10}+x^{12} \rangle
\end{aligned}$ &&$[16,4^4 2^0,12]^{\ast\ast}$&&\\
\hline
\hline
\caption{Free linear codes over $\ZZ_4$}\label{Tabel-1}
\end{longtable}
\end{center}
\vspace{-10ex}
Notes for Table \ref{Tabel-1}, Table \ref{Tabel-2}, Example \ref{Ex-final-1}, and Example \ref{Ex-final-2}:
\begin{itemize}
    \item $^\ast$ means that the code has new $k_1$ and $k_2,$ but there is/are other codes of equal length in the database \cite{Aydin} with the same minimum Lee distance but with greater cardinalities.

    \item $^{\ast \ast}$ means that the code has minimum Lee distance greater than all existing linear codes of equal length with the same values of $k_1$ and $k_2$ in the database \cite{Aydin}.

    \item $^{\ast \ast \ast}$ means that the code has new $k_1$ and $k_2,$ with greater or equal cardinalities compared with all existing linear codes of equal length with the same value of minimum Lee distance in the database \cite{Aydin}.

    \item $^{\ast)}$ means that the code has the same parameters as some existing good linear codes of equal length in the database \cite{Aydin}.

	\item $C^{PS}$ is the code obtained by Plotkin-sum construction, namely $C^{PS}:=\{(\xv~|~\xv+\yv):\xv,\yv \in Res(C)\}$ or $C^{PS}:=\{(\xv~|~\xv+\yv):\xv,\yv \in Tor(C)\}.$
\end{itemize}
In the Table \ref{Tabel-1} above, notice that all linear codes over $\ZZ_4$ constructed by this method are free. This is not a coincidence. In fact, this is a direct consequence of Theorem \ref{T-Gray-1}, as follows.
\begin{cor}\label{C-Res}
	Let $C$ be a free $\Delta_\theta$-cyclic code of length $n$ over $R$ with $|C|=16^k.$  Then the free linear codes $Res(C)$ and $Tor(C)$ have parameter $4^k2^0,$ and the free linear code $\Phi(C)$ has parameter $4^{2k}2^0.$
\end{cor}

\begin{rem}
The part of Corollary \ref{C-Res} related to the parameters of $Res(C)$ and $Tor(C)$ does not hold for linear codes over $\ZZ_4+u\ZZ_4,$ with $u^2=1$ considered by Sharma and Bhaintwal \cite{Sharma2018} (e.g., see Example 7 in \cite{Sharma2018}).\end{rem}

\subsection{Notes on a computational simplification}

From Corollary \ref{C-Res}, we conclude that we need non-free linear codes over $R$ to obtain non-free Gray images, residue codes, and torsion codes. Here, we use a slightly modified construction from principally generated $\Delta_\theta$-cyclic codes. Notice that for any $g(x) \in R_{n,\Delta_{\theta}},$ the set $\{g(x),xg(x),\ldots,x^{n-1}g(x)\}$ is a generating set, which is not necessary minimal, of the (not necessarily free) $\Delta_{\theta}$-cyclic code $C=\langle g(x) \rangle$ over $R.$  We %as well as Sharma and Bhainthwal %\cite{Sharma2018}
consider the subcode $C_k=\langle g(x),xg(x),\ldots,x^{k-1}g(x)\rangle$ of $C$ for some $k\leq n$. %\textcolor{red}{(sebenarnya saya kurang yakin metode apa yang digunakan di %paper tsb karena tidak dijelaskan secara eksplisit)}
From there, we can obtain $Res(C_k),Tor(C_k),$ and $\Phi(C_k).$  For the case of $\Delta_{\theta}$-cyclic codes over $R,$  we can simplify our computation as follows.

Let $C$ be a $\Delta_{\theta}$-cyclic code over $R$ of length $n$ generated by $\{g(x),xg(x),\ldots,x^{n-1}g(x)\},$ where $g(x)=c_0+c_1x+c_2x^2+\cdots+c_{n-1}x^{n-1} \in R_{n,\Delta_{\theta}}$ and $c_i=a_i+b_iv,$ for $0 \leq i \leq n-1.$  To be more precise, for an even and an odd $n,$ the generator of $C \seq R^n$ consists of the following vectors, respectively:
\[
\begin{aligned}
	g_1&=(a_0+b_0v,~a_1+b_1v,~\ldots,~a_{n-2}+b_{n-2}v,~a_{n-1}+b_{n-1}v),\\
	g_2&=(a_{n-1}+b_{n-1}+b_0-b_{n-1}v,~a_0+b_0+b_1-b_0v,~\ldots,~a_{n-3}+
	b_{n-3}+b_{n-2}-b_{n-3}v,\\
	   &\quad~a_{n-2}+b_{n-2}+b_{n-1}-b_{n-2}v),\\
	g_3&=(a_{n-2}+b_{n-2}v,~a_{n-1}+b_{n-1}v,~\ldots,~a_{n-4}+b_{n-4}v,~a_{n-3}+b_{n-3}v),\\
	   &\quad\vdots\\
	g_n&=(a_1+b_1+b_2-b_1v,~a_2+b_2+b_3-b_2v,~\ldots,~a_{n-1}+
	b_{n-1}+b_0-b_{n-1}v,~a_0+b_0+b_1-b_0v),
\end{aligned}
\]
and
\[
\begin{aligned}
	g_1&=(a_0+b_0v,~a_1+b_1v,~\ldots,~a_{n-2}+b_{n-2}v,~a_{n-1}+b_{n-1}v),\\
	g_2&=(a_{n-1}+b_{n-1}+b_0-b_{n-1}v,~a_0+b_0+b_1-b_0v,~\ldots,~a_{n-3}+
	b_{n-3}+b_{n-2}-b_{n-3}v,\\
	&\quad~a_{n-2}+b_{n-2}+b_{n-1}-b_{n-2}v),\\
	g_3&=(a_{n-2}+b_{n-2}v,~a_{n-1}+b_{n-1}v,~\ldots,~a_{n-4}+b_{n-4}v,~a_{n-3}+b_{n-3}v),\\
	&\quad\vdots\\
	g_n&=(a_1+b_1v,~a_2+b_2v,~\ldots,~a_{n-1}+b_{n-1}v,~a_0+b_0v).
\end{aligned}
\]
By applying the Lemma \ref{Res-1}, we obtain the generator of $Res(C)$ for an even and an odd $n$ that consists of the following vectors, respectively:
\[
\begin{aligned}
	Res(g_1)&=(a_0,~a_1,~a_2,~\ldots,~a_{n-2},~a_{n-1}),\\
	Res(g_2)&=(a_{n-1}+b_{n-1}+b_0,~a_0+b_0+b_1,~\ldots,~a_{n-3}+
	b_{n-3}+b_{n-2},~a_{n-2}+b_{n-2}+b_{n-1}),\\
	Res(g_3)&=(a_{n-2},~a_{n-1},~\ldots,~a_{n-4},~a_{n-3}),\\
	&\quad\vdots\\
	Res(g_n)&=(a_1+b_1+b_2,~a_2+b_2+b_3,~\ldots,~a_{n-1}+
	b_{n-1}+b_0,~a_0+b_0+b_1),
\end{aligned}
\]
and
\[
\begin{aligned}
	Res(g_1)&=(a_0,~a_1,~a_2,~\ldots,~a_{n-2},~a_{n-1}),\\
	Res(g_2)&=(a_{n-1}+b_{n-1}+b_0,~a_0+b_0+b_1,~\ldots,~a_{n-3}+
	b_{n-3}+b_{n-2},~a_{n-2}+b_{n-2}+b_{n-1}),\\
	Res(g_3)&=(a_{n-2},~a_{n-1},~\ldots,~a_{n-4},~a_{n-3}),\\
	&\quad\vdots\\
	Res(g_n)&=(a_1,~a_2,~\ldots,~a_{n-1},~a_0).
\end{aligned}
\]
Moreover, we obtain the generator of $Tor(C)$ for an even and an odd $n$ that consists of the following vectors, respectively:
\[
\begin{aligned}
	Tor(g_1)&=(a_0+b_0,~a_1+b_1,~a_2+b_2,~\ldots,~a_{n-2}+b_{n-2},~a_{n-1}+b_{n-1}),\\
	Tor(g_2)&=(a_{n-1}+b_0,~a_0+b_1,~\ldots,~a_{n-3}+
	b_{n-2},~a_{n-2}+b_{n-1}),\\
	Tor(g_3)&=(a_{n-2}+b_{n-2},~a_{n-1}+b_{n-1},~\ldots,~a_{n-4}+b_{n-4},~a_{n-3}+b_{n-3}),\\
	&\quad\vdots\\
	Tor(g_n)&=(a_1+b_2,~a_2+b_3,~\ldots,~a_{n-1}+b_0,~a_0+b_1),
\end{aligned}
\]
and
\[
\begin{aligned}
	Tor(g_1)&=(a_0+b_0,~a_1+b_1,~a_2+b_2,~\ldots,~a_{n-2}+b_{n-2},~a_{n-1}+b_{n-1}),\\
	Tor(g_2)&=(a_{n-1}+b_0,~a_0+b_1,~\ldots,~a_{n-3}+
	b_{n-2},~a_{n-2}+b_{n-1}),\\
	Tor(g_3)&=(a_{n-2}+b_{n-2},~a_{n-1}+b_{n-1},~\ldots,~a_{n-4}+b_{n-4},~a_{n-3}+b_{n-3}),\\
	&\quad\vdots\\
	Tor(g_n)&=(a_1+b_1,~a_2+b_2,~\ldots,~a_{n-1}+b_{n-1},~a_0+b_0).
\end{aligned}
\]
For any $k \leq n,$ if $C_k$ is generated by the set $\{g_1,g_2,\ldots,g_k\},$ then the $Res(C_k)$ and $Tor(C_k)$ is generated by $\{Res(g_1),Res(g_2),\ldots,Res(g_k)\}$ and $\{Tor(g_1),Tor(g_2),\ldots,Tor(g_k)\},$ respectively. This method can reduce the computation time significantly. Using this method, we can obtain some linear codes over $\ZZ_4$ with good parameters, as listed in the table below.
\begin{center}
\begin{longtable}{c|c|c|c}
\hline\hline
{\multirow{2}{*}{Generators of $C$}} %& $\Phi(C)$
& $Res(C)$ & $Tor(C)$ & $C^{PS}$ \\
\cline{2-4}
%& $[n,4^{k_1}2^{k_2},d_L]$
& $[n,4^{k_1}2^{k_2},d_L]$ & $[n,4^{k_1}2^{k_2},d_L]$ & $[n,4^{k_1}2^{k_2},d_L]$ \\
\hline

		$\{g_1(x),xg_1(x)\}$ &  $[4,4^1 2^1,4]^\ast$ && $[8,4^2 2^2,4]$ \\
		\hline
		$\{g_2(x),xg_2(x),x^2g_2(x)\}$ &  $[4,4^1 2^2,2]^\ast$ && $[8,4^2 2^4,2]$ \\
		\hline
		$\{g_3(x),xg_3(x),x^2g_3(x)\}$ &  $[5,4^2 2^1,4]^{\ast\ast\ast}$ & &  \\
		\hline
		$\{g_4(x),xg_4(x),x^2g_4(x)\}$ &  $[6,4^2 2^1,4]^\ast$ & & \\
		
		\hline
		$\{g_5(x),xg_5(x),x^2g_5(x),x^3g_5(x),x^4g_5(x)\}$ & $[6,4^2 2^3,4]^{\ast)}$ & & $[12,4^4 2^6,4]^\ast$ \\
		
		\hline
		$\{g_6(x),xg_6(x),x^2g_6(x),x^3g_6(x)\}$ & $[6,4^3 2^1,4]^\ast$ & &$[12,4^6 2^2,4]$ \\
		
		\hline
		$\begin{aligned}
			\{g_7(x),xg_7(x),x^2g_7(x),x^3g_7(x),\\
			x^4g_7(x),x^5g_7(x), x^6 g_7(x)\}
		\end{aligned}$ & $[8,4^4 2^3,4]^{\ast)}$ & & $[16,4^{8} 2^6,4]^\ast$\\
		\hline
		$\begin{aligned}
			\{g_8(x),xg_8(x),x^2g_8(x),x^3g_8(x),\\
			x^4g_8(x),x^5g_8(x)\}
		\end{aligned}
		$ &  & $[8,4^5 2^1,4]^{\ast\ast\ast}$ &$[16,4^{10} 2^2,4]^\ast$ \\
		\hline
		$\{g_9(x),xg_9(x),x^2g_9(x)\}$ & $[9,4^3 2^0,7]^{\ast)}$ &   \\
		\hline
		$\{g_{10}(x),xg_{10}(x),x^2g_{10}(x)\}$  & $[10,4^3 2^0,8]^{\ast)}$ &   \\
		\hline
		$\{g_{11}(x),xg_{11}(x)\}$  & $[15,4^2 2^0,15]^{\ast\ast}$ &  \\
		\hline
		$\{g_{12}(x),xg_{12}(x),x^2g_{12}(x)\}$ & & $[18,4^3 2^0,14]$ & \\
		\hline
		\hline
		\caption{Linear codes over $\ZZ_4$}\label{Tabel-2}
	\end{longtable}
\end{center}
\vspace{-12ex}
In Table \ref{Tabel-2},
\begin{itemize}
	\item $g_1(x)=(1+3v)+2x+(3+3v)x^2.$

	\item $g_2(x)=(1+v)+(2+2v)x+(1+3v)x^2.$
	
	\item $g_3(x)=(1+3v)+2vx+(2+2v)x^2+2vx^3+(1+3v)x^4.$
	
	\item $g_4(x)=3+(1+3v)x+(3+v)x^2+(2+3v)x^3.$
	
	\item $g_5(x)=(3+3v)+(1+3v)x+(3+3v)x^3+(3+v)x^4+2x^5.$
	
	\item $g_6(x)=(1+v)+x+(2+v)x^2+vx^3+3vx^5.$
	
	\item $g_7(x)=(1+v)+ 3 x+ (2+3v)x^2+ (3+v)x^3+ (2+2v) x^4+ 2 x^6+ x^7.$
	
	\item $g_8(x)=2v+ (2+3v) x+ (1+3v) x^2+ (1+2v) x^3+ 2v x^4+ (1+v) x^5+ x^6+ 3v x^7.$
	
	\item $g_9(x)=(1+v)+ (1+3v) x+ (3+2v) x^2+ 3x^3+ (3+3v) x^4+ (2+2v) x^5+ (2+3v) x^6+ (3+3v) x^7+ (1+3v) x^8.$
	
	\item $g_{10}(x)=(1+3v)+ (2+2v) x+ 3 x^2+ v x^4+ 3 x^5+ (3+3v) x^6+  x^7+  x^8.$
	
	\item $g_{11}(x)=(1+v)+ (2+v) x^2+ (3+2v) x^3+ (3+v) x^5+ (1+2v) x^6+ 2 x^7+ (3+v) x^8+ 3 x^9+ 2v x^{10}+ (3+2v) x^{11}+ (1+2v) x^{12}+ (2+v) x^{13}+ (1+3v) x^{14}.$
	
	\item $g_{12}(x)=1+ 2v x+ (3+2v) x^2+ (2+2v) x^3+ (1+2v) x^4+ x^5+ v x^6+ x^7+ (1+2v) x^8+ (3+2v) x^9+ (1+3v) x^{10}+ (3+3v) x^{11}+ (2+v) x^{12}+ (3+3v) x^{13}+  x^{14}+ (2+v) x^{15}+ (2+2v) x^{16}.$
\end{itemize}

Exactly the same method can also be used for $\Phi(C_k).$  This observation brings us to the conclusion that the codes $Res(C_k),$  $Tor(C_k),$ and $\Phi(C_k)$ have at most $4^k,~4^k$%\textcolor{red}{(lebih umumnya k1+k2<=k)},
, and $4^{2k}$ %\textcolor{red}{(lebih umumnya k1+k2<=2k)}
codewords,  respectively.  Moreover, the similar observation can also be applied to the codes over $\ZZ_4+u\ZZ_4,$ with $u^2=1,$ investigated by Sharma and Bhaintwal \cite{Sharma2018} and it is easy to verify that in this case, the codes
$Res(C_k),$  $Tor(C_k),$ and $\Phi(C_k)$ have at most $4^{2k}$ codewords.

\subsection{Construction from another method}

By applying the construction in \cite{Hopein2022}, we can obtain even more new linear codes over $\ZZ_4$ with the highest known minimum Lee distance, previously unknown to exist in the database \cite{Aydin}. As illustrations, we provide several examples below.  In the following examples, the code $C_i$ refers to the code generated by the polynomial $g_i(x)$ in Table \ref{Tabel-2}.

\begin{ex}
	In these examples, we use Lemma 4.6 in \cite{Hopein2022} to construct new linear codes over $\ZZ_4.$
\begin{itemize}
	\item From the code with parameters $[6,4^2 2^0, 6]$ in Table \ref{Tabel-1}, we obtained linear codes with parameters $[12, 4^2 2^1, 12]^{*)},$ $[24, 4^3 2^0, 24]^{*)},$ $[24, 4^2 2^2, 24]^{*)},$ $[48, 4^3 2^1, 48]^{***},$ and $[48, 4^2 2^3, 48]^{*)}.$

	\item From the code with parameters $[12,4^3 2^0, 10]$ in Table \ref{Tabel-1}, we obtained linear codes with parameters $[24, 4^3 2^1, 20],$ $[48, 4^4 2^0, 40],$ and $[48, 4^3 2^2, 40]^{***}.$
		% (keterangan: + artinya $k_1, k_2, d_L$ baru, tidak ada pembanding di database).

	\item From the code $C_3$ with parameters $[5, 4^2 2^1, 4]$ we obtained linear codes with parameters $[10, 4^2 2^2, 8]^{***},$ $[20, 4^3 2^1, 16],$ $[20, 4^2 2^3, 16]^{*},$ $[40, 4^3 2^2, 32]^{*},$ and $[40, 4^2 2^4, 32].$

	\item From the code $C_5$ with parameters $[6, 4^2 2^3, 4]$ we obtained a linear code with parameters $[12, 4^2 2^4, 8]^{**}.$

	\item From the code $C_6$ with parameters $[6, 4^3 2^1, 4]$ we obtained a linear code with parameters $[12, 4^3 2^2, 8]^*.$

	\item From the code $C_9$ with parameters $[9,4^3 2^0, 7]$ we obtained linear codes with parameters $[18, 4^3 2^1, 14]^{**},$ $[36, 4^4 2^0, 28]^{**},$ and $[36, 4^3 2^2, 28]^{**}.$

	\item From the code $C_{10}$ with parameters $[10,4^3 2^0, 8]$ we obtained a linear code with parameters $[40, 4^4 2^0, 32].$
	\end{itemize} \label{Ex-final-1}
\end{ex}
\begin{ex}
In these examples, we use Lemma 4.3 in \cite{Hopein2022} to construct new linear codes over $\ZZ_4.$
\begin{itemize}
	\item From the code with parameters $[12,4^3 2^0, 10]$ in Table \ref{Tabel-1} and the code with parameters $[24, 4^3 2^0, 24],$ respectively, we obtained a linear code with parameters $[36, 4^3 2^0, 34]^{*)}.$
	
	\item From the code $C_3$ with parameters $[5, 4^2 2^1, 4]$ and the code with parameters $[12, 4^2 2^1, 12]$ we obtained a linear code with parameters $[17, 4^2 2^1, 16]^{*}.$
	
	\item From the code $C_6$ with parameters $[6, 4^3 2^1, 4]$ and the code with parameters $[20, 4^3 2^1, 16]$ we obtained a linear code with parameters $[26, 4^3 2^1, 20]^{*}.$
\end{itemize} \label{Ex-final-2}
\end{ex}

%See \cite{Hopein2022} for the detail.
%This result, which is now in preparation will be published elsewhere in a separate paper \cite{Hopein2021+}.

\section{Concluding remarks}
We have investigated the algebraic structures of skew-cyclic codes, also known as $\theta$-cyclic codes, with a derivation $\Delta_{\theta}$ over the ring $R=\ZZ_4+v\ZZ_4,$ with $v^2=v,$ extending the observation of Boucher and Ulmer \cite{Boucher2014}, where they defined and considered the skew-cyclic codes with derivation over a finite field.  To our best knowledge, this is the third attempt after the paper by Sharma and Bhaintwal \cite{Sharma2018} and Patel and Prakash \cite{Patel2022}.   
%Our work can also be viewed as a partial answer of the question of Dougherty \cite{Dougherty2017}, page 99:

%\noindent{\bf Question 6.5}  Determine the structure of all $\theta$-cyclic codes in $R[x;\theta]/ \langle x^n-\lambda\rangle$ where $R$ is a finite commutative Frobeni\"{u}s ring.

As a consequence, we constructed several new codes over $\ZZ_4$ unknown to exist before due to the database \cite{Aydin}, with good parameters.  All computations to find the codes were done by Python and Magma computational algebra system \cite{Magma}.

% While investigating the new good codes over $\ZZ_4,$ we were aware that for the case of linear codes over $\ZZ_4,$ we have the following upper bound for the minimum Lee distance

% \begin{equation}\label{LeeBound}
% d_L \leq \frac{|C|}{|C|-1}n.
% \end{equation}

% The bound above has been found in 1968 by Wyner and Graham \cite{Wyner1968}, but we are able to reprove the bound in an elementary way.  In many cases, the bound is much better than the known Lee distance bound derived by Dougherty and Shiromoto \cite{Dougherty2001} as given by Theorem \ref{Lee-bound} below.

% \begin{theorem}[Lee distance bound]\label{Lee-bound}
% 	If $C$ is a linear code of length $n$ over $\ZZ_4$ with parameters $[n,4^{k_1}2^{k_2},d_L],$ then $d_L \leq 2n-2k_1-k_2+1.$
% \end{theorem}

Regarding the derivation, it is easy to show that the map $\Delta_{\theta}(x)=(3+2v)(\theta(x)-x)$ also defines a derivation on $R.$  All properties which hold for the derivation $\Delta_{\theta}(x)=(1+2v)(\theta(x)-x)$ in this paper also hold for the derivation $\Delta_{\theta}(x)=(3+2v)(\theta(x)-x).$

The method explained in Section \ref{Application} can be modified for the case of derivation $\Delta_{\theta}(x)=(3+2v)(\theta(x)-x).$  There is some hope to obtain many more examples of linear codes over $\ZZ_4$ with better parameters.  As an example, $Tor(C)$ of the code $C:=\langle g(x),xg(x),x^2g(x) \rangle,$ with
$g(x)=3v+(2+v)x+3vx^2+vx^3,$ has parameters $[4,4^1 2^2,4],$ which is better than the one in Table \ref{Tabel-1}.

\section*{Acknowledgments}
%The authors thank three anonymous referees for their meticulous reading of the manuscript.  Their suggestions have been very valuable in improving the presentation of this paper.

This research is supported by the Institut Teknologi Bandung (ITB) and the Ministry of Education, Culture, Research and Technology
(\emph{Kementerian Pendidikan, Kebudayaan, Riset dan Teknologi (Kemdik\-budristek)}),
Republic of Indonesia.

%\section*{Declarations}

%\noindent{\bf Conflict of interest: } Not applicable.

\end{document}